\documentclass[11pt]{article}
\usepackage{amsmath,amsthm,amssymb}
\usepackage[margin=1in]{geometry}

\usepackage{tikz}
\usepackage{algorithm,algpseudocode}
\algtext*{EndIf}
\algtext*{EndFor}

\newtheorem{theorem}{Theorem}
\newtheorem{corollary}{Corollary}
\newtheorem{lemma}{Lemma}
\newtheorem{claim}{Claim}
\theoremstyle{definition}
\newtheorem{definition}{Definition}
\newtheorem{construction}{Construction}

\newcommand{\pluseq}{\mathrel{+}=}

\title{Relaxation for Efficient Asynchronous Queues}

\author{Samuel Baldwin, Cole Hausman, Mohamed Bakr,\\Edward Talmage\footnote{Corresponding Author} (elt006@bucknell.edu)}

\date{Bucknell University, Lewisburg, PA, USA}

\begin{document}

\maketitle

\begin{abstract}
We explore the problem of efficiently implementing shared data structures in an asynchronous computing environment.  We start with a traditional FIFO queue, showing that full replication is possible with a delay of only a single round-trip message between invocation and response of each operation.  This is optimal, or near-optimal, runtime for the $Dequeue$ operation.  We then consider ways to circumvent this limitation on performance.  Though we cannot improve the worst-case time per operation instance, we show that \emph{relaxation}, weakening the ordering guarantees of the Queue data type, allows most $Dequeue$ instances to return after only local computation, giving a low amortized cost per instance.  This performance is tunable, giving a customizable tradeoff between the ordering of data and the speed of access.

\textbf{Keywords:} Distributed Data Structures \and Asynchronous Algorithms \and Relaxed Data Types
\end{abstract}

\section{Introduction}

Shared data is a core aspect of modern computing.  Datasets are growing, local computational power cannot keep up, and data collection is continually growing more widespread.  In the naturally resulting geographically distributed systems, one of the primary concerns is accessing data that other members of the system may be accessing concurrently.  Without careful coordination, we can inadvertently use stale versions of the data, overwrite other users' work, or even corrupt the entire data structure.

Distributed data structures are a fundamental construct for allowing computing to spread across different machines which need to work together.  By specifying the interface of data operations users may call and the exact effects of those operations, we provide an abstraction layer that allows a programmer designing for a distributed system to not worry about the details of coordinating and maintaining data.  Instead, they can focus on their application.  In co-located parallel computation, hardware can provide shared access to memory, but in geographically distributed systems, we must build that abstraction layer, considering all the possible oddities of concurrency, messages, delays, and different views of the data.  We focus on optimizing the interactions with these system elements, and provide the developer with an efficient tool which provides well-specified guarantees.

In this work, we provide distributed, message-passing algorithms implementing shared data types.  Our goal is to minimize the time between when a user invokes an operation on a simulated shared data structure and when that operation returns.  Past work has given tight or nearly tight bounds on implementations of a variety of data structures \cite{Kosa99,WangTalmageLeeWelch18} in partially synchronous systems, where processes know bounds on the real time which messages between participating processes may take to arrive.  However, real-world systems do not generally have reliable bounds on timings, so these algorithms are difficult to implement in practical systems.  As a next step towards such practical implementations of efficient distributed data structures, we here consider an asynchronous model of computation.  In this model, processes cannot rely on any local measure of time, or on message delays carrying any system-level information about the timing or order of remote operations.  Instead, we construct logical timestamps and use them to agree on the order in which operations appear to take place.

This paper contains two primary results.  First, we present an algorithm for a standard FIFO queue in a fully asynchronous message passing model.  To our knowledge, this is the first published fully-replicated, integrated queue algorithm for an asynchronous system.  This algorithm is of some, but limited, independent interest, since there are extant algorithms with similar performance.  In Section 17.1.2 of her textbook \cite{Lynch96}, Lynch presents a general construction for implementing shared data structures in message passing systems which can be instantiated to implement a queue.  That algorithm is centralized, with one controlling process storing all data and determining the linearization of operation instances.  Our solution matches the performance of this algorithm, but in a fully-distributed way, with all processes equivalent in the system.  Lubowich and Taubenfeld \cite{LubowichTaubenfeld11} gave a similar implementation, though in a modular form relying on separate locking and counter implementations and without formal analysis.  That algorithm also only partially replicated data across processes.  

Replication is important to us for a few reasons.  First, in our second result we use replication to improve performance.  We are able to have some $Dequeue$ instances return after only local work, not waiting for communication, but this is only possible when the relevant data is stored at the invoking process, which is not necessarily true in the prior algorithms.  Second, looking ahead to future work, one of the primary uses of data replication is fault tolerance.  In a centralized algorithm, if the coordinating process crashes, the entire system fails.  With data replication, it may be possible to continue to operate in the presence of crashes.  This paper does not handle fault tolerance, but that is our intended next step, after our work in this paper to tolerate asynchrony.  Thus, adding replication at this point puts us in a better position to build on this work in the future.

Our second result considers the notion of \emph{relaxing} a data type: weakening the guarantees of the data type specification \cite{HenzingerKirschPayerSezginSokolova13}.  We can then use this weakening to achieve better performance than for a basic data type.  Extending an algorithm of Talmage and Welch \cite{TalmageWelch14}, we consider a version of a Queue datatype that allows $Dequeue$ to not necessarily return the single oldest value, but any sufficiently old value in the structure.  This allows us to determine return values for most $Dequeue$ invocations without waiting for any communication, thus bringing the amortized cost of $Dequeue$ far below the minimum possible for an unrelaxed queue.  We provide an algorithm doing this, building on our first algorithm.  This algorithm is the first to implement a relaxed queue in an asynchronous system and proves that the performance benefits relaxation gave in partially synchronous systems are also possible in asynchronous systems.  We even reduce the cost of inexpensive $Dequeue$s below what that paper achieved in a better-behaved system.  While the relaxed queue does not provide the same ordering guarantees as a FIFO queue, in many concurrent executions when multiple processes are concurrently dequeueing, the behavior will be indistinguishable.  If a process dequeues, for example, the third oldest element in the queue, this appears similar (at least at that time), to an execution of an unrelaxed queue where two other processes concurrently dequeued the two older elements, as this process will still dequeue the third oldest element.

Another intriguing property of relaxed queues is that they are not as subject to the impossibility results for asynchronous, fault-tolerant systems \cite{ShavitTaubenfeld16,TalmageWelch19}.  Thus, it may well be possible to extend our result in the future to asynchronous, fault-tolerant implementations.  This would be a major step forward, since the impossibility of implementing queues and most other useful data types \cite{Herlihy91,FischerLynchPaterson85} severely limits the capabilities of distributed systems.  This paper is a step on the road toward that goal, and we are excited to pursue such implementations.

\subsection{Related Work}

There is a large body of work on message-passing implementations of shared data structures in partially-synchronous models, where it is easier to prove lower bounds much more easily than in real-world systems (e.g. \cite{Kosa99,WangTalmageLeeWelch18,ChandraHadzilacosToueg16}).  In these models, we know that complex operations like $Dequeue$ on a queue are inherently expensive, needing to wait for communication before they can return to the user.  Simpler operations, such as $Read$ and $Write$ operations on a register can be faster, but often the sum of their worst-case delays cannot be too low \cite{MavronicolasRoth99}.

We build on the algorithms of Wang et al. \cite{WangTalmageLeeWelch18} and Talmage and Welch \cite{TalmageWelch14}.  Our first algorithm extends that in \cite{WangTalmageLeeWelch18} with vector clocks to handle asynchrony.  Our second algorithm similarly extends one from \cite{TalmageWelch14}, using the construction from our first algorithm to handle asynchrony.  Wade and Talmage \cite{WadeTalmage20} showed that the relaxed queue algorithm we build on does not need full replication, but can store a single copy of each data element, without giving up too much performance.  For clarity, we do not include that optimization in this work, but it should be straightforward to modify our algorithm in a similar way.

In asynchronous models, there is a great deal of active research on very strong primitives such as consensus objects, particularly in the presence of failures in the system (e.g. \cite{PuFarahbakhshAlvisiEyal23,CohenKeidar23}).  These objects can implement arbitrary other objects \cite{Herlihy91}, so there is good reason to pursue them, but their extra strength can make them more expensive to implement, and they are typically not deterministic, since asynchronous fault-tolerant consensus is deterministically impossible \cite{FischerLynchPaterson85}.  At the other end of the spectrum, simple types like registers have a long history of message-passing implementations, largely centered on the classic ABD algorithm \cite{AttiyaBar-NoyDolev95}.  This algorithm tolerates both asynchrony and process failures.  However, registers are limited in their ability to implement more complex data types.  Our goal is to implement a somewhat more complex data type as efficiently as possible.  We also do not yet consider failures, though that is the next step in our research trajectory.

Any concurrent implementation of a sequentially-specified data type must satisfy a \emph{consistency condition}, which specifies how concurrent behavior maps to sequences of operations, and defines which concurrent behaviors are allowed.  Linearizability, introduced by Herlihy and Wing \cite{HerlihyWing90} is the standard condition, giving both the most intuitive behavior and providing other useful guarantees like composability of different shared objects, but coming at a high time cost \cite{AttiyaWelch94}.  There is much work exploring weaker consistency conditions (See Viotti and Vukolich's work for a survey \cite{viottiVukolic16}), which have a variety of different properties and possible performance levels.

Another approach to circumventing some of the lower bounds on linearizable implementations is to weaken the data type itself.  This concept, known as \emph{relaxation}, was pioneered by Afek et al. \cite{AfekKorlandYanovsky10}, then formalized by Henzinger et al. \cite{HenzingerKirschPayerSezginSokolova13}.  Talmage and Welch proved some lower bounds, but showed that better amortized time for some types of relaxed $Dequeue$ is possible than for an unrelaxed $Dequeue$ \cite{TalmageWelch14}.  We follow that work, but in an asynchronous model, while those algorithms assumed knowledge of message delay bounds.

While we focus primarily on performance, another interesting aspect of our work is the computability of relaxed data types.  Herlihy \cite{Herlihy91} established consensus number, the maximum number of processes which can use a type to solve distributed consensus in an asynchronous system subject to crash failures, as a measure of the computational power of a data type.  Shavit and Taubenfeld \cite{ShavitTaubenfeld16} and Talmage and Welch \cite{TalmageWelch19} showed that relaxing a queue can, in some cases, reduce its computational power. While this may seem undesirable, it could perhaps allow implementations of structures that behave very similarly to a traditional queue but which are not subject to the impossibility of solving consensus.  In fact, Casta{\~n}eda, Rajsbaum, and Raynal \cite{CastanedaRajsbaumRaynal20} demonstrated another relaxation of a queue, called multiplicity, that allows non-blocking set-linearizable implementation (a similar but slightly weaker type of implementation than that considered in consensus numbers) from $Read$/$Write$ registers, a significant step towards implementing relaxed queues that tolerate asynchrony and failures, while behaving in predictable ways.  We hope as a next step to add failure-tolerance to the algorithms we present in this paper, while maintaining low operation delays.

\section{Model and Definitions}

\subsection{Asynchronous System Model}
We assume a standard fully asynchronous message passing model of computation.  The system contains a set of $n$ processes $\Pi = [p_0, \dots , p_{n-1}]$ modeled as state machines, and a set of $n$ users, one for each process.  State machines accept two types of input event: the corresponding user may \emph{invoke} an operation, or the machine may \emph{receive} a message.  Each of these will trigger a handler.  Each process, in its handlers, may perform local computation or one of of two external actions: it may \emph{respond} to its user or \emph{send} a message to another process.  The state machines are time-free, meaning that their output is only described through their input and state transitions without any upper or lower time bound on computation or message delay.  A user may not invoke an operation at a particular process until the process responds to its last invocation.  

We consider failure-free systems, in processes always behave per the state machine specification.  We assume all inter-process communication is reliable, so when a process sends a message, that message will arrive at its destination process exactly once after a finite amount of time.  Additionally, we assume communication channels between processes are First-In, First-Out (FIFO).  That is, if process $p_i$ sends message $m_1$ to process $p_j$, then sends message $m_2$ also to $p_j$, $p_j$ will receive $m_1$ before it receives $m_2$.  This assumption does not reduce the generality of our results, as we can implement such an ordering guarantee by attaching a sequence number to each message and buffering incoming messages at each process.  We store out of order messages in this buffer until all previous messages are received, then receive the buffered message.

While time is unavailable to processes in the system, we establish a notion of the time cost of our algorithms by considering that they run in real time.  When analyzing our algorithms, we will speak of a \emph{run}, which comprises a sequence of state-machine events for each process.  We consider \emph{timed runs}, which are simply runs with a real-time associated with each event.  A run is \emph{admissible} if every message arrives exactly once, in FIFO order, users invoke only only operation instance at a time, and each state machine runs infinitely.

There are no upper or lower bounds on either local computation time or the time between when a message is sent and received.  However, we can measure an algorithmic time, from outside the system, by expressing it in terms of the number of sequential messages which must occur in that duration.  That is, in a particular run, up to a finite point where the algorithm reports that it terminates, we define the parameter $d$ as the real time duration of the longest time between sending and receiving any one message.  We express algorithmic time in terms of $d$, by considering how many messages the algorithm sends and receives where a send causally must occur after the receive of the previous message \cite{Awerbuch85}.  Since each of those delays may equal the largest in that run, this gives a measure of algorithmic runtime without bounding real time message delays.  See Section 2.1.2 in \cite{AttiyaWelch94} or Section 3.2.1 in \cite{Barbosa96} for more details.

\subsection{Data Type Definitions}

We first state the definition for a standard FIFO queue abstract data type, then that for the relaxed queue we implement.  We state abstract date type (ADT) specifications in two parts: First, a list of operations the user can invoke, with argument and return types, expressed as $OP(arg,ret)$.  We use $-$ to indicate when a function takes no argument or returns no value.  Second, we define the set of legal sequences of \emph{instances} of these operations.  An operation instance is an invocation-response pair, noting that these are separate events in a distributed setting.

To simplify the statement of our definitions, we assume that each argument to $Enqueue$ is unique.  One way to achieve this practically would be a simple abstraction layer that adds a (logical) timestamp to each value before it is passed to the data structure.  We also define the notion of \emph{matching}: A $Dequeue$ instance matches an $Enqueue$ instance if it returns that $Enqueue$ instance's (unique) argument.  We will use the special character $\bot$ to represent an empty queue.

\begin{definition}\label{def:FIFOQueue}
  A \emph{Queue} over a set of values V is a data type with two operations:
  \begin{itemize}
  \item $Enqueue(val,-), val \in V$ 
  \item $Dequeue(-, val), val \in V \cup \{\bot\}$ 
  \end{itemize}
  
  The empty sequence is legal.  For any legal sequence $\rho$ of instances of queue operations and $val \in V$, (1) $\rho \cdot Enqueue(val,-)$ is legal, (2) $\rho \cdot Dequeue(-,val)$ is legal iff $Enqueue(val, -)$ is the first unmatched $Enqueue$ instance in $\rho$, and (3) $\rho \cdot Dequeue(-, \bot)$ is legal iff every $Enqueue(val, -)$ in $\rho$ is matched.
\end{definition}

We can now formally define the relaxed queue we implement in this paper.  Intuitively, each $Dequeue$ instance can return one of the $k$ oldest elements in the queue.  If $k=1$, this is the same as the FIFO Queue defined above.

\begin{definition} A \emph{Queue with $k$-Out-of-Order relaxed $Dequeue$}, or a \emph{$k$-Out-of-Order Queue}, over a set of values $V$ is a data type with two operations:
  \begin{itemize}
  \item $Enqueue(val,-), val \in V$
  \item $Dequeue(-,val), val \in V$
  \end{itemize}

  The empty sequence is legal.  For any legal sequence $\rho$ of instances of $k$-out-of-order queue operations and $val \in V$, (1) $\rho \cdot Enqueue(val,-)$ is legal, (2) $\rho \cdot Dequeue(-,val)$ is legal iff $val$ is the argument of one of the first $k$ unmatched $Enqueue$ instances in $\rho$, and (3) $\rho \cdot Dequeue(-,\bot)$ is legal iff there are fewer than $k$ unmatched $Enqueue$ instances in $\rho$.
\end{definition}

We are interested in implementations which satisfy \emph{linearizability} \cite{HerlihyWing90}.  Linearizability is a consistency condition which describes how concurrent executions of a data structure relate to the sequential specification of the ADT the structure implements.  Specifically, linearizability requires that in any admissible timed run, operation response values be such that there is a total order of all operation instances in the execution which is legal according to the ADT and which respects the real-time order of instances which do not overlap in real time.  That is, if instance $op_1$ returns before, in real time, instance $op_2$'s invocation, $op_1$ must precede $op_2$ in the order.  This condition provides behavior that matches our intuitive expectations of a system, disallowing inversions in real time.  Linearizability also has the advantage of being composable, meaning that running multiple linearizable objects in the same execution will necessarily produce an overall linearizable execution.  Linearizability is the strongest consistency condition, so our upper bounds translate to any other consistency condition one might use.

\section{Asynchronous FIFO Queue Algorithm}

\subsection{Description}

We present a fully-replicated implementation of a queue, with coordination done by logical timestamping and messages announcing and confirming invocations.  This algorithm is based on Talmage and Welch's relaxed queue algorithm \cite{TalmageWelch14}, but updated to replace that algorithm's use of timers to agree on a linearization order with vector clocks \cite{Mattern88,Fidge91,SchwarzMattern94}.  This, along with acknowledgment of invocation messages allows us to tolerate asynchrony.  When a user invokes an operation at a particular process, that process increments its local vector clock and then broadcasts that timestamp with a request for the operation instance to take effect.  It then waits until it receives confirmations from all other processes, then returns to the user.  This delay guarantees that all other processes have updated their timestamps and that the invoking process is aware of all operation instances with smaller timestamps.  This enables linearization by guaranteeing that concurrent operation instances are long enough to be aware of each other and return appropriate values.

When a process receives an operation invocation, with its timestamp and any applicable argument values, it updates its vector clock to be larger than that of the received invocation and send an acknowledgment back to the invoker.  For $Dequeue$ instances, the process also sends acknowledgments to all other processes, to enable them to apply all $Dequeue$ instances to their local replica in the correct order.  Depending on the type of the operation, the acknowledging process will either apply it to its own replica immediately on receiving an acknowledgment, in the case of an $Enqueue$ instance, or after receiving acknowledgments from all processes, in the case of a $Dequeue$ instance.  We prove below that this delay is sufficient to guarantee that processes can locally execute every $Dequeue$ instance in timestamp order and enables us to use that order for our linearization.  Since we do not delay local execution of $Enqueue$ instances, we must be careful not to return the argument of an $Enqueue$ to a $Dequeue$ instance too soon, but since we are using timestamp order and associate timestamps with values in the local replicas, this is straightforward.

\paragraph*{Vector Clocks}  Each process maintains a logical vector clock that holds its local view of the number of steps of the algorithm (invocation and message send) each process has taken \cite{Mattern88,Fidge91,SchwarzMattern94}.  We refer to each $p_i$'s vector clock as $v_i$.  Each process' vector clock is an array of size $n$ that is initially 0 at all indices.  When any process $p_i$ invokes $Enqueue$ or $Dequeue$, it increments index $i$ in its local clock, $v_i[i]$, and broadcasts the new vector as the invocation's timestamp.  When another process receives a message containing a timestamp, it will update its vector clock value by first incrementing its own component, then updating the value at each index of its clock vector to the larger of its previous value and the corresponding value in the received timestamp.  This guarantees that each component of the local clock will be at least as large as the received timestamp, recording its knowledge of the previous remote event.

We define two orders on timestamp vectors.  For any two unequal vector timestamps $v_i$ and $v_j$, we say $v_i$ is \emph{strictly smaller} than $v_j$, denoted $v_i \prec  v_j$, if $v_i[x] \leq v_j[x],\forall x \in [0, \dots, n-1]$.  If this is not true, let $k$ be the index where the vectors first differ.  That is, for $x = 0$ to $k-1$, $v_i[x] = v_j[x]$, but $v_i[k] \leq v_j [k]$.  Then we say that $v_i$ is \emph{lexicographically smaller} than $v_j$ and write $v_i << v_j$.  Notice that $v_i \prec v_j$ implies that $v_i << v_j$ but not the opposite.

\paragraph*{Confirmation Lists}
The main algorithmic innovation that allows us to handle asynchrony is a structure we call \emph{Confirmation Lists}.  Each process uses these lists to track acknowledgments from other processes for a given $Dequeue$ instance.  When a process $p_j$ learns about a $Dequeue$ invocation at process $p_i$, $p_j$ will create a \emph{confirmation list} object containing the timestamp that uniquely identifies the invocation and an array of $n$ Booleans to track responses from each process.  The process inserts that object in a list of confirmation lists, $PendingDequeues$, which is sorted in increasing lexicographic order of timestamps, using the function $PendingDequeues.insertByTS(confList)$.

A process can learn about an invocation by receiving a $Dequeue$ request message directly from the invoking process or by receiving a third-party process' response message to a $Dequeue$ invocation.  A confirmation list's $responses$ array initially contains $False$ at every index, except for the position corresponding to the invoking process, which is $True$.  The list fills as the process receives acknowledgments to the $Dequeue$ request, setting the position corresponding to the responding process to $True$.

When a process $p_j$ receives a response from process $p_k$ for a $Dequeue$ instance $op$, it knows that $p_k$ has updated its local clock to be strictly larger than $op$'s timestamp, and thus necessarily larger than the timestamp of any other $Dequeue$ instance with a lexicographically smaller timestamp.  $p_j$ will thus set not only the value at index $k$ of $op$'s confirmation list to $True$, but the value at index $k$ in any other pending $Dequeue$ instances with smaller timestamps.  When the confirmation list at $p_j$ for $op$ contains only $True$ values in its response array, it has received responses from all processes so $p_j$ can locally execute that $Dequeue$ instance, applying it to its own local copy of the queue.  

\paragraph*{Local Replica}
Each process maintains a variable $lQueue$, which is its local replica of the simulated shared queue.  We implement $lQueue$ as a minimum priority queue storing values, process ids, and timestamps, keyed on timestamp.  Each entry is the argument of an $Enqueue$ instance, with the invoking process' id and timestamp.  We remove elements in priority order, but allow the constraint of passing a parameter which upper-bounds the timestamp of the removed value.  If there is no value in $lQueue$ with timestamp lexicographically smaller than that bound, $lQueue$ will indicate that it is empty by returning $\bot$.  We invoke these operations as follows:
\begin{itemize}
\item $lQueue.insertByTS(val, inv, ts)$: Adds an entry containing the triple $(val, inv, ts)$ to $lQueue$, sorting by lexicographic timestamp order on $ts$.
\item $lQueue.dequeue(ts)$: Returns and removes the oldest element in $lQueue$ with timestamp smaller than $ts$, $\bot$ if there is none.
\end{itemize}

\subsection{Pseudocode}

Our algorithm modifies that of Talmage and Welch \cite{TalmageWelch14}.  Instead of using timers based on knowledge of message delays, we use round-trip messages to indicate when processes can apply operation instances to their local replicas of the shared queue.  Pseudocode appears in Algorithms~\ref{alg:fifo} and \ref{alg:helpers}.

\begin{algorithm}[!ht]
  \caption{Code for each process $p_i$ to implement a Queue}\label{alg:fifo}
  \begin{algorithmic}[1]
    \Function{Enqueue}{$val$}\label{fifoline:invEnq}
      \State $EnqResponseCount = 0$ \Comment{Count responses to this invocation}
      \State $Enq_{ts} = updateTS()$ \label{fifoline:enqTS} \Comment{Increment local vector clock, read instance's timestamp}
      \State send $(EnqReq, val, Enq_{ts}, i)$ to all processes\label{fifoline:sendEnqReq}
    \EndFunction

    \Function{Receive}{$EnqReq, val, Enq_{ts}, inv$} from $p_{inv}$
      \State $updateTS(Enq_{ts})$ \label{fifoline:enqReqTSUpdate} \Comment{Update local vector clock by invocation's timestamp}
      \State $lQueue.insertByTS(val, inv, Enq_{ts})$\label{fifoline:executeEnq} \Comment{Locally execute the $Enqueue$ instance}
      \State send $(EnqAck, i)$ to $p_{inv}$ \label{fifoline:sendEnqAck} \Comment{Acknowledge receipt of the invocation}
    \EndFunction

    \Function{Receive}{$EnqAck, j$ from $p_j$}
      \State $EnqResponseCount \pluseq 1$
      \If {$EnqResponseCount == n$} 
        \Return $EnqResponse$ to user\label{fifoline:enqReturn} 
      \EndIf
    \EndFunction \label{fifoline:finishEnq}
    \Function{Dequeue}{}
      \State $Deq_{ts} = updateTS()$ \Comment{Increment local vector clock, read instance's timestamp}\label{fifoline:deqTS}
      \State send $(DeqReq, Deq_{ts}, i)$ to all processes \label{fifoline:sendDeqReq}
    \EndFunction
    \Function{Receive}{$(DeqReq, Deq_{ts}, inv)$ from $p_{inv}$}
    \State $updateTS(Deq_{ts})$ \label{fifoline:deqReqTSUpdate}
      \If{$Deq_{ts}$ is not in $PendingDequeues$}
      \State $PendingDequeues.insertByTS(createList(Deq, Deq_{ts}, p_{inv}$))\label{fifoline:savePendingDeq}
      \EndIf
      \State send $(DeqAck, Deq_{ts}, p_{inv}, i)$ to all processes \label{fifoline:sendSafetyFlag}\label{fifoline:sendDeqAck}
    \EndFunction
    \Function{Receive}{$(DeqAck, Deq_{ts}, p_{inv}, j)$ from $p_j$}
      \If{$Deq_{ts}$ not in $PendingDequeues$}
        \State $PendingDequeues.insertByTS(createList(Deq, Deq_{ts},p_{inv}))$
      \EndIf
      \State Let $currentConfList$ be the confirmation list in $PendingDequeues$ with $ts == Deq_{ts}$
      \State $currentConfList.responses[j] = True$ \label{fifoline:setResponse}
      \State $propagateEarlierResponses(PendingDequeues, Deq_{ts})$ \label{fifoline:propagateEarlier}

      \For{each $confirmationList$ in $PendingDequeues$, in increasing lexicographic timestamp order}
        \If{all $response$s in $confirmationList$ are $True$}\label{fifoline:fullConfList}
          \State $ret = lQueue.dequeue(Deq_{ts})$\label{fifoline:chooseDeqValue} \Comment{Locally execute the $Dequeue$ instance} 
        \State delete $confirmationList$ from $PendingDequeues$
        \If{$p_i == p_{inv}$} \Return $ret$ to user \label{fifoline:deqReturn} \Comment{Invoking process responds to user} \EndIf
      \EndIf 
      \EndFor

      \EndFunction
      \Statex \emph{(continues below)}

    \algstore{fifoLines}  
  \end{algorithmic}
\end{algorithm}

\begin{algorithm}[t]
  \caption{Algorithm~\ref{alg:fifo} continued: Helper functions}\label{alg:helpers}
  \begin{algorithmic}[1]
    \algrestore{fifoLines}
    
    \Function{propagateEarlierResponses}{$PendingDequeues, Deq_{ts}$}
      \Statex \Comment{If a later-timestamp instance receives a $DeqAck$ from $p_j$, then any earlier-timestamp instances need wait no longer for a $DeqAck$ from $p_j$.}
      \For{each confirmation list $PendingDequeues[k]$ starting at $ts = Deq_{ts}$, in decreasing timestamp order} 
        \For{each process id $j$}
          \If{$PendingDequeues[k].responses[j]$ and $k>0$}
            \State $PendingDequeues[k-1].responses[j] = True$
          \EndIf
        \EndFor
      \EndFor  
    \EndFunction
    \Function{updateTS}{$v_j$}
      \State $v_i[i] \pluseq 1$
      \If{$v_j$ is not empty}
        \For{$k = 0$ to $n - 1$}
          \State $v_i[k] = max(v_i[k], v_j[k])$
        \EndFor
      \EndIf
      \State return $v_i$
    \EndFunction
    \Function{createList}{$op, Deq_{ts}, p_{invoker}$}
      \State create a new confirmation list $confirmationList$
      \State $confirmationList.op = op$
      \State $confirmationList.responses = [False, False,...,False]$ \Comment{An array of $n$ Booleans}
      \State $confirmationList.responses[invoker] = True$
      \State $confirmationList.ts = Deq_{ts}$ \Comment{A vector clock timestamp}
      \State $confirmationList.invoker = p_{invoker}$ \Comment{The invoking process of the Dequeue}
      \State return $confirmationList$
    \EndFunction
  \end{algorithmic}
\end{algorithm}

\subsection{Correctness}\label{sec:fifoCorrectness}
In order to prove that this algorithm correctly implements the queue specification, consider an arbitrary, admissible timed run $R$.  We will first prove that all invocations have a matching return, so that we have complete operation instances.  We can then construct a total order of all operation instances in $R$ and prove that it respects real time order, making it a valid linearization, and is legal by the specification of a queue.  The core of the proof is in proving that when any process removes a particular $Dequeue$ instance's return value from its local replica of the queue, it deletes the same value as every other process does.  This means that all replicas will undergo the same series of operations, and allows us to prove that they continue removing the correct values.  

\begin{lemma}\label{fifolem:responses}
  In $R$, every invocation has a matching response.
\end{lemma}

\begin{proof}
When a user invokes an operation at process $p_i$, the algorithm sends a message containing the invocation to all processes on line~\ref{fifoline:sendEnqReq} or \ref{fifoline:sendDeqReq}, depending on the operation invoked.  Each process will receive that message in finite time by our assumption of reliable channels and send back an $EnqAck$ for an $Enqueue$ invocation on line~\ref{fifoline:sendEnqAck} or a $DeqAck$ for a $Dequeue$ invocation on line~\ref{fifoline:sendSafetyFlag}, as appropriate.  Each of those responses will arrive in finite time, and when $p_i$ receives all $n$ of them, it will generate a matching return on line~\ref{fifoline:enqReturn} for $Enqueue$ or line~\ref{fifoline:deqReturn} for $Dequeue$.
\end{proof}

Each process reads its vector clock at each invocation, on line~\ref{fifoline:enqTS} or \ref{fifoline:deqTS}, and associates that clock value with that invocation throughout the algorithm.  We will refer to this clock value as the \emph{timestamp} of the instance containing that invocation.

\begin{construction}\label{constr:fifo}
  Let our linearization $\pi$ be the increasing lexicographic timestamp order of all operation instances in $R$.
\end{construction}

\begin{lemma}\label{fifolem:realTimeOrder}
  $\pi$ respects the real time order of non-overlapping instances.
\end{lemma}
\begin{proof}
  Proof by contradiction. Let there be two non-overlapping operation instances $op_1$ and $op_2$, where $op_1$ returns prior to $op_2$'s invocation in real time.  Assume, for the sake of contradiction, that $op_2$ is prior to $op_1$ in $\pi$.
  
  Given that $op_1$ occurs before $op_2$ in real time, and that the two operation instances are non-overlapping, $op_2$'s timestamp must be larger, lexicographically, than $op_1$'s.  This follows from the fact that $op_1$ will not return until it receives a response from every process, including $op_2$'s invoking process (lines~\ref{fifoline:enqReturn}, \ref{fifoline:deqReturn}.  This means that $op_2$'s invoking process received the invocation for $op_1$ and updated its timestamp (line~\ref{fifoline:enqReqTSUpdate} or \ref{fifoline:deqReqTSUpdate}) before $op_1$ returned, and thus before it invoked $op_2$, so $op_2$'s timestamp will be totally ordered after $op_1$'s, implying it is also lexicographically after.  Construction~\ref{constr:fifo} then places $op_1$ before $op_2$ in $\pi$, contradicting our assumption.
\end{proof}

Next, we consider how each process locally executes each $Dequeue$ instance.  We show that all processes do so in the same way, meaning that their replicas of the data structure move through similar sequences of states, and thus continue to behave properly.  Processes may locally execute $Enqueue$ instances at different times, but since we store those instances' timestamps with the values in the local replicas, no $Dequeue$ instance never delete a value from an $Enqueue$ instance later in the linearization.

\begin{lemma}\label{fifolem:prevLocalExec}
  When a process $p_i$ locally executes any $Dequeue$ instance $op$, for any instance $op'$ with a lexicographically smaller timestamp than $op$, $p_i$ has previously locally executed $op'$.
\end{lemma}

\begin{proof}
  By line~\ref{fifoline:fullConfList}, $p_i$ will not locally execute $op$ until it has received a $DeqAck$ from every process, either for $op$ or for a $Dequeue$ instance with larger timestamp than $op$ (line~\ref{fifoline:propagateEarlier}).  But when any process sends a $DeqAck$, on line~\ref{fifoline:sendDeqAck}, it has updated its timestamp to be larger than $op$'s timestamp.  Thus, any instance invoked at that process after that point will have a strictly, and thus lexicographically, larger timestamp than $op$.  This means that any instance $op'$ invoked at $p_j$ with a lexicographically smaller timestamp than $op$ must have been invoked before its invoking process sent a $DeqAck$ for $op$ or any later $Dequeue$ instance.

  By FIFO message order, $p_j$'s request message for $op'$, which it sends on line~\ref{fifoline:sendEnqReq} or \ref{fifoline:sendDeqReq}, would have arrived at $p_i$ before the $DeqAck$ for $op$ from $p_j$, and $p_i$ would have either locally executed $op'$ on line~\ref{fifoline:executeEnq} if it was an $Enqueue$ instance, or put it in $PendingDequeues$ on line~\ref{fifoline:savePendingDeq} if it was a $Dequeue$ instance.  If it was a $Dequeue$ instance, the only time $p_i$ could remove $op'$ from $PendingDequeues$ is in lines~\ref{fifoline:fullConfList}-\ref{fifoline:deqReturn}, when $p_i$ locally executes $op'$.  Thus, when $p_i$ locally executes $op$, it must have already either locally executed $op'$, or have it in $PendingDequeues$.

  But $p_i$ can only locally execute $op$ when its confirmation list is full, and any time $p_i$ sets a response entry in $op$'s confirmation list to true on line~\ref{fifoline:setResponse}, it then propagates responses to all pending $Dequeue$ instances with smaller timestamps on line~\ref{fifoline:propagateEarlier}.  Thus, if $p_i$ has received a response from every process for $op$, it has also marked every $response$ cell in the confirmation list for $op'$, and will locally execute $op'$ before $op$.
\end{proof}

This means that during local execution of a $Dequeue$ instance, each process has already placed the arguments of all earlier-linearized $Enqueue$ instances in its replica of the queue and has already locally executed every $Dequeue$ instance with a smaller timestamp.  

\begin{corollary}\label{fifolem:localExecOrder}
  Each process locally executes all $Dequeue$ instances in increasing lexicographic timestamp order.
\end{corollary}

\begin{lemma}\label{fifolem:legal}
  $\pi$ is a legal sequence by the specification of a FIFO queue.
\end{lemma}

\begin{proof}
  We prove this claim by induction on $\pi$.  The empty sequence is legal.  Assume that in an arbitrary prefix $\rho \cdot op$ of $\pi$, $\rho$ is legal.  We will show that $rho \cdot op$ is also legal which, by induction, implies that $\pi$ is legal.  Let $p_i$ be the process which invoked $op$.

  If $op$ is an $Enqueue$ instance, then $\rho \cdot op$ is legal by Definition~\ref{def:FIFOQueue}.

  If $op$ is a $Dequeue$ instance, $op = Dequeue(-,ret)$, we must show that either $ret \neq \bot$ is the argument of the first unmatched $Enqueue$ instance in $\rho$ or there $ret = \bot$ and there are no unmatched $Enqueue$ instances in $\rho$.

  We begin with the case that $ret \neq \bot$.  This means that when $p_i$ chose $ret$ as the return value in line~\ref{fifoline:chooseDeqValue}, it found $ret$ in its local replica of the queue, with timestamp lexicographically less than $ts(op)$.  Thus, $ret$ was the argument of some $Enqueue$ instance in $\pi$, since the only values added to an $lQueue$ are $Enqueue$ arguments, in line~\ref{fifoline:executeEnq}.

  We next prove that $Enqueue(ret,-)$ is the first unmatched $Enqueue$ instance in $\rho$.  By Lemma~\ref{fifolem:prevLocalExec}, when $p_i$ locally executes $op$ and chooses its return value, it has already locally executed every $Enqueue$ instance with a lexicographically smaller timestamp.  By Construction~\ref{constr:fifo}, these are the $Enqueue$ instances in $\rho$.  Thus, the set of elements $p_i$ has available to choose are the arguments of $Enqueue$ instances in $\rho$.  Further, by Corollary~\ref{fifolem:localExecOrder}, $p_i$ has locally executed every $Dequeue$ instance with timestamps smaller than $op$, which are exactly those $Dequeue$ instances in $\rho$.

  Thus, $p_i$'s $lQueue$ contains the arguments of unmatched $Enqueue$ instances in $\rho$, as well as potentially the arguments of some $Enqueue$ instances which appear in $\pi$ after $op$.  Line~\ref{fifoline:chooseDeqValue} chooses the value in $lQueue$ with the smallest timestamp, filtering to only those with timestamp smaller than $op$, which is the argument of the first unmatched $Enqueue$ instance in $\rho$, since we constructed $\rho$ in timestamp order.  Similarly, when any other process locally executes $op$, they will delete the same value from their $lQueue$ by the same logic, simply not returning it to their user by the check on line~\ref{fifoline:deqReturn}.  

  The next case we consider is when $ret = \bot$.  In this case, when $p_i$ executes line~\ref{fifoline:chooseDeqValue}, it finds no element in $lQueue$ with smaller timestamp than $op$.  By the logic in the previous case, this means that every $Enqueue$ instance with timestamp smaller than $op$ is matched by a $Dequeue$ instance in $\rho$, so $\rho \cdot op$ is legal.

  Thus, by induction, $\pi$ is legal.
\end{proof}

\begin{theorem}
  Algorithm~\ref{alg:fifo} implements a FIFO queue.
\end{theorem}

\begin{proof}
  By Lemma~\ref{fifolem:realTimeOrder}, $\pi$ respects the real-time order of non-overlapping instances and by Lemma~\ref{fifolem:legal}, $\pi$ is a legal sequence of operation instance on a FIFO queue.  Thus, every admissible timed run $R$ of Algorithm~\ref{alg:fifo} has a valid linearization, so it is a correct implementation.
\end{proof}

\subsection{Complexity}

As we proved in Lemma~\ref{fifolem:responses}, the algorithm returns to every invocation after it receives acknowledgments from every other process.  This takes up to $2d$ time, since the request message must reach each other process before it can send an acknowledgment, and each of these messages can take up to $d$ time.  This is equivalent to the best existing algorithm \cite{Lynch96}.

The lower bound from Wang et al. \cite{WangTalmageLeeWelch18} for pair-free operations applies here, since we use a less restricted model.  This bound implies that the worst-case cost of $Dequeue$ in an asynchronous model must be at least $4d/3$.  However, in an asynchronous system, a process has no way to know when a certain amount of time has elapsed.  The only possible timing knowledge is that when a message round trip completes, at most $2d$ time has elapsed.  Thus, our algorithm's runtime appears to be optimal.

\section{Asynchronous Out-of-Order Queues}

\subsection{Description}

We next move on to consider how we might be able to improve performance, despite the lower bound on $Dequeue$.  Relaxation has previously been useful for circumventing such lower bounds \cite{TalmageWelch14}, so we consider it here.  Similarly as for Algorithm~\ref{alg:fifo}, we update that existing algorithm to handle asynchrony, using relaxation to improve the common-case cost of $Dequeue$ and give better overall performance.  We split $Dequeue$ instances into two types: ``slow'' $Dequeue$s that behave like those in the FIFO algorithm above, taking $2d$ time to return, and ``fast'' $Dequeue$s which return a value immediately on execution, not waiting for the coordinating communication that slows down other operations.  To enable the algorithm to do this without violating legality, we partition the set of $k$ values which a $Dequeue$ instance can return and give each process ownership of a disjoint subset.  Then that process can safely $Dequeue$ any of its own elements locally, knowing that no other process will return them to a $Dequeue$.  When a process runs out of elements, its next $Dequeue$ instance will be slow, as it takes the time to coordinate with other processes to determine its return value.  We use that same time to assign ownership of more values, so that future $Dequeue$ instances at that process can be fast again.

The conversion to an asynchronous model is algorithmically similar to that for an unrelaxed queue, but the correctness proof is more involved.  We can no longer simply use timestamp order from our vector clocks to construct a linearization.  Since fast $Dequeue$ instances take no time at all, they must linearize when they happen in real time, which may not match logical timestamps.  Instead, we must insert each fast $Dequeue$ instance into our linearization and prove that its return value was, at that moment, one of the $k$ oldest values in the queue.  We must also prove that our construction respects the real time order of non-overlapping instances, which was much more straightforward for the unrelaxed queue.

Overall, if $k \geq n$, this means that at least half of $Dequeue$ instances can take no time at all for communication delay, vastly outperforming the unrelaxed queue.  If $k < n$, our algorithm falls back to the above algorithm for unrelaxed queues.  The worst case per-instance cost is identical between the algorithms, but the average performance with relaxation is significantly better.  One interesting aspect of this algorithm, that makes it significantly more complex than just adding relaxation handling as previously done in partially synchronous systems to Algorithm~\ref{alg:fifo} is that we can no longer rely solely on timestamp order for linearization.  Since fast $Dequeue$ instances return instantly, there is no guarantee that timestamp order correctly captures their real-time order.  Thus, we must construct a much more complex linearization order to satisfy real-time order, then prove that it is legal.  We rely on invariants related to ownership of elements to ensure that fast $Dequeue$ instances choose values that will be legal for them to return at their linearization points, even if those points are quite different than the timestamp order the code actually sees. 

\subsection{Algorithm}

Because this algorithm is built on Algorithm~\ref{alg:fifo}, we omit some portions that are identical to those in that algorithm.  For example, the code for handling $Enqueue$ instances is exactly the same as that on lines~\ref{fifoline:invEnq}-\ref{fifoline:finishEnq} of Algorithm~\ref{alg:fifo}, so we omit it here, as well as the helper functions.  The code for handling $Dequeue$ is similar to that above, but differentiates fast and slow instances, so we present it here. 

In addition to the functions we used before, we add the following functions to the variable storing each process' replica of the shared structure, which help with labeling elements for fast $Dequeue$:
\begin{itemize}
\item $lQueue.peekByLabel(p)$: Return, without removing, the oldest element in $lQueue$ labeled for process $p$, $\bot$ if there is none.
\item $lQueue.deqByLabel(p)$: Return and remove the oldest element in $lQueue$ labeled for process $p$.
\item $lQueue.deqUnlabeled(ts)$: Return and remove the oldest unlabeled element in $lQueue$ which has timestamp less than $ts$.
\item $lQueue.remove(val)$: Remove the specific value $val$ from $lQueue$.
\item $lQueue.unlabeledSize()$: Return the number of unlabeled elements in $lQueue$.
\item $lQueue.labelOldest(p,x)$: Label the $x$ oldest unlabeled elements in $lQueue$ for process $p$.
\end{itemize}
  Pseudocode appears in Algorithm~\ref{alg:relaxed}.

\begin{algorithm}
  \caption{Code for each process $p_i$ to implement a queue with $k$-Out-of-Order relaxed $Dequeue$.  Helper functions are as in Algorithm~\ref{alg:helpers}}\label{alg:relaxed}
  \begin{algorithmic}[1]
    \Function{Dequeue}{}
      \State $Deq_{ts} = updateTS()$
      \If {$localQueue.peekByLabel(p_{i}) \neq \bot$}\label{oooline:checkFast}
        \State $ret = localQueue.deqByLabel(p_i)$ \label{oooline:fastDeq}
        \State send $(Deq_f, ret, Deq_{ts}, i)$ to all processes
        \State \Return $ret$ to user \label{oooline:fastDeqResponse}
      \Else { } send $(Deq_s, null, Deq_{ts}, i)$ to all processes
      \EndIf
    \EndFunction

    \Function{Receive}{($op, val, Deq_{ts}, inv)$ from $p_{inv}$}
      \State $updateTS(Deq_{ts})$
      \If{$Deq_{ts}$ is not in $PendingDequeues$}
        \State $PendingDequeues.insertByTS(createList(op, Deq_{ts}, p_{inv}$))
      \EndIf
      \State send $(op, val, DeqAck, Deq_{ts}, p_{inv}, i)$ to all processes
    \EndFunction

    \Function{Receive}{$(op, val, DeqAck, Deq_{ts}, p_{inv}, j)$ from $p_j$}
      \If{$Deq_{ts}$ not in $PendingDequeues$}
        \State $PendingDequeues.insertByTS(createList(op, Deq_{ts}, p_{inv}))$
      \EndIf
      \State Let $currentConfList$ be the confirmation list in $PendingDequeues$ with $ts == Deq_{ts}$
      \State $currentConfList.responses[j] = True$
      \State $propagateEarlierResponses(PendingDequeues, Deq_{ts})$ 
      \For{each $confirmationList$ in $PendingDequeues$, in increasing lexicographic timestamp order}
        \If{all $responses$ in $confirmationList$ are $True$}\label{oooline:localExec}
          \If {$confirmationList.op == Deq_f$}
            \If{$p_i \neq p_{inv}$} $lQueue.remove(val)$\EndIf
          \Else    
              \State $ret = lQueue.deqUnlabeled(Deq_{ts})$ \label{oooline:sDeqChooseUnlabeled}
            \State $labelElements(p_{inv})$\label{oooline:label}
            \If{$p_i == p_{inv}$} \Return $ret$ to user
            \EndIf
          \EndIf 
        \EndIf 
      \EndFor
      \EndFunction
      \Function{labelElements}{$p_j$} from \cite{TalmageWelch14}
      \State $y = lQueue.unlabeledSize()$
      \State $lQueue.labelOldest(p_j,x)$, where $x = \min\{\lfloor k/n\rfloor, y\}$\label{oooline:labelOldest}
      \EndFunction
  \end{algorithmic}
\end{algorithm}

\subsection{Correctness}
The primary difference in the correctness proof for Algorithm~\ref{alg:relaxed} compare to that for the unrelaxed queue, is that the linearization must be more complex, to deal with fast $Dequeue$ instances.  Several lemmas from Section~\ref{sec:fifoCorrectness} still apply to $Enqueue$ and slow $Dequeue$ instances, though we need to extend them to account for fast $Dequeue$ instances.  Once we construct our linearization, the out of order behavior of these instances does not significantly change the correctness argument, it only expands our access to the structure.

Formally, a fast $Dequeue$ instance is one which passes the check on line~\ref{oooline:checkFast}.  Any $Dequeue$ instance which fails this check is slow.

Let $R$ be an arbitrary, admissible, timed run of Algorithm~\ref{alg:relaxed}.

\begin{lemma}
 In $R$, every invocation has a matching response.
\end{lemma}

\begin{proof}
  $Enqueue$ and slow $Dequeue$ instances are exactly as in Algorithm~\ref{alg:fifo}, so we omit the proof here, referring the reader to Lemma~\ref{fifolem:responses}.  Fast $Dequeue$ instances will necessarily generate a return on line~\ref{oooline:fastDeqResponse}.  Thus every invocation has a matching response.
\end{proof}

Since we now have operation instances, we can construct a permutation of all instances in $R$.  We then prove that this permutation is a valid linearization, which shows the correctness of our algorithm.  We define operation instance timestamps as for the previous algorithm.

\begin{construction}\label{constr:relaxed}
  Construct an order of all operation instances in $R$ as follows:
  \begin{enumerate}
  \item Let $\pi$ be the increasing lexicographic order of all $Enqueue$ and slow $Dequeue$ instances. \label{orderSlow}
  \item For each fast $Dequeue$ $op$, in increasing real-time invocation order, let
    \begin{itemize}
    \item $\rho$ be the subsequence of all $Enqueue$ and slow $Dequeue$ instances which return before $op$'s invocation;
    \item $\sigma$ be the subsequence of all fast $Dequeue$ instances already in $\pi$ which return before $op$'s invocation;
    \item $\tau$ be the subsequence of all $Enqueue$ and slow $Dequeue$ instances invoked after $op$ returns.
    \end{itemize}
    Place $op$ in $\pi$ immediately after the later of the last elements of $\rho$ and $\sigma$.\label{orderFast}
  \end{enumerate}
\end{construction}

We need to prove that this sequence respects the real-time order of all non-overlapping operation instances.  We first note that the subsequence defined in step~\ref{orderSlow} respects real-time order by the argument in Lemma~\ref{fifolem:realTimeOrder}.

\begin{lemma}\label{ooolem:fastDeqOrdering}
  When we place each fast $Dequeue$ instance $op$ in step~\ref{orderFast} of Construction~\ref{constr:relaxed}, it precedes, in $\pi$, all elements of $\tau$.
\end{lemma}

\begin{proof}
  First, note that all elements of $\rho$ precede all elements of $\tau$ in $\pi$, by step~\ref{orderSlow}.  Next, we prove that all elements of $\sigma$ also precede all elements of $\tau$ in $\pi$.

  Since every element of $\tau$ is invoked after $op$ returns and every element of $\sigma$ returns before $op$ is invoked, every element of $\sigma$ returns before any element of $\tau$ is invoked.  Thus, an element $op'$ of $\sigma$ only follows an element of $\tau$ in $\pi$ if some previously-placed fast $Dequeue$ instance (in the $\sigma'$ defined when we placed $op'$) was placed in $\pi$ later than that element of $\tau$.  This argument holds in turn for that previously-placed fast $Dequeue$ instance, repeating all the way back to the first fast $Dequeue$ instance which Step~\ref{orderFast} placed in $\pi$.  That could only be placed after an instance in $op$'s $\tau$ if a previous one was, but there was no previous one, so neither that first fast $Dequeue$ nor any subsequent one could be placed in $\pi$ after an element of $\tau$.  Thus, $op$ precedes every element of $\tau$ in $\pi$.
\end{proof}

\begin{lemma}
  Construction~\ref{constr:relaxed} respects the real-time order of non-overlapping instances in $R$.
\end{lemma}

\begin{proof}
  Consider any two operation instances $op_1$ and $op_2$ in $R$ which do not overlap in real time.  We break the proof down by cases for whether each of $op_1$ and $op_2$ is a fast $Dequeue$ instance or a slow instance: an $Enqueue$ or slow $Dequeue$.
  \begin{itemize}
  \item Both $op_1$ and $op_2$ are slow instances.  This follows from Lemma~\ref{fifolem:realTimeOrder}.
  \item Both $op_1$ and $op_2$ are fast $Dequeue$ instances.  Order follows by construction.  We place the later, in real time order, of $op_1$ and $op_2$ second, and by the definition of $\sigma$ in step~\ref{orderFast} of Construction~\ref{constr:relaxed}, $op_1$ is in $\sigma$, so we place $op_2$ in $\pi$ after $op_1$.
  \item One of $op_1$ and $op_2$ is a fast $Dequeue$ instance, the other is a slow instance.  When we place the fast instance (WLOG, $op_1$) in $\pi$, since $op_1$ and $op_2$ do not overlap, then if $op_2$ precedes $op_1$ in real-time order, it was thus in $\rho$, so we place $op_1$ after $op_2$ in $\pi$ .  If $op_1$ precedes $op_2$ in real-time order, then $op_2$ is in $\tau$, and Lemma~\ref{ooolem:fastDeqOrdering} proves that $op_1$ precedes $op_2$ in $\pi$.
  \end{itemize}

  Thus, in all cases, $\pi$ respects the real-time order of any pair of instances $op_1$ and $op_2$.
\end{proof}

Now that we have a linearization of all instances in our run, we need to show that the return values the algorithm chooses are legal.  $Enqueue$ and slow $Dequeue$ instances behave as in the unrelaxed case, with the only difference being how a slow $Dequeue$ chooses its return value.  Now, it will skip elements labeled for other processes and take the oldest unlabeled element.  Since we label at most $\lfloor n/k\rfloor$ elements per process and a $Dequeue$ is only slow when its invoking process has no labeled elements, this will still leave it choosing from among the $k$ oldest elements in the queue, which is legal by the definition of the relaxed queue.  We begin our legality proof by restating Lemma~\ref{fifolem:prevLocalExec} and its corollary in the context of the related queue implementation.  We do not re-prove it, as the proof is fundamentally identical.  We note that, despite fast $Dequeue$ instances removing their return value from the invoking process' $lQueue$ in line~\ref{oooline:fastDeq}, we still refer to local execution of the instance as the code starting when line~\ref{oooline:localExec} passes.

\begin{lemma}
  When a process $p_i$ locally executes any $Dequeue$ instance $op$, for any instance $op'$ with a lexicographically smaller timestamp than $op$, $p_i$ has previously locally executed $op'$.
\end{lemma}

\begin{corollary}\label{ooolem:localExecOrder}
  Each process locally executes all $Dequeue$ instances in increasing lexicographic timestamp order.
\end{corollary}

We can now prove the following invariant of our labeling scheme that will enable us to argue that all $Dequeue$ instances return legal values.

\begin{lemma}\label{ooolem:labelOwner}
  Once line~\ref{oooline:labelOldest} labels an element $x$ for $p_i$, no other process will return $x$ to a $Dequeue$ it invoked.
\end{lemma}

\begin{proof}
  The algorithm chooses $Dequeue$ return values in two places: line~\ref{oooline:fastDeq} at a fast $Dequeue$'s invocation or on line~\ref{oooline:sDeqChooseUnlabeled} in local execution of a slow $Dequeue$ at its invoking process.  The latter of these will never return a labeled element, so the claim vacuously holds there.  In the first option, it chooses one labeled for the instance's invoking process, and we have the claim.
\end{proof}

We can further observe that the code never removes a label, so once processes label a particular element for $p_i$, they will only remove it as the return value of a $Dequeue$ instance invoked at $p_i$.

\begin{lemma}
  $\pi$ is a legal sequence of operation instances, by the definition of a queue with $k$-Out-of-Order relaxed $Dequeue$.
\end{lemma}

\begin{proof}
  We proceed by induction on $\pi$ and prove the following claim in tandem with and to aid in the primary result:

  \begin{claim}\label{ooolem:safeLabels}
    Any $Dequeue$ instance $op$ which labels elements during its local execution (line~\ref{oooline:label}) labels only elements which are among the arguments of the $k$ oldest $Enqueue$ instances in the prefix of $\pi$ ending with that $Dequeue$ instance.  All processes label the same elements during local execution of $op$.
  \end{claim}

  The empty sequence is legal.  We assume that for an arbitrary prefix $\rho \cdot op$ of $\pi$, $\rho$ is legal and Claim~\ref{ooolem:safeLabels} holds for all instances in $\rho$.  We will show that $\rho \cdot op$ is also legal and any element $op$ labels is the argument to one of the first $k$ unmatched $Enqueue$ instances in $\rho \cdot op$.  Let $p_i$ be $op$'s invoking process and proceed by cases on $op$'s operation type:
  \begin{itemize}
  \item $op$ is an $Enqueue$ instance.  Since $Enqueue$ has no return value, $\rho \cdot op$ is always legal.
  \item $op$ is a slow $Dequeue$ instance, $op = Dequeue(-,ret)$.  $p_i$ chose $ret$ from its $lQueue$ in line~\ref{oooline:sDeqChooseUnlabeled}.  
    If $ret \neq \bot$, then we know that it chose the argument of an unmatched $Enqueue$ instance in $\rho$, since the algorithm puts only $Enqueue$ arguments in $lQueue$, line~\ref{oooline:sDeqChooseUnlabeled} only returns values with timestamps below $ts(op)$, and Corollary~\ref{ooolem:localExecOrder} implies that $p_i$ has already locally executed any $Dequeue$ instance in $\rho$, removing its return value from $lQueue$.  Further, $deqUnlabeled$ chooses the oldest eligible element currently in $lQueue$.  By line~\ref{oooline:labelOldest}, no more than $\lfloor k/n\rfloor$ elements are labeled for each process and there are no elements labeled $p_i$, or $op$ would be a fast $Dequeue$.  Thus, there are fewer than $k$ labeled elements in $lQueue$ and the oldest unlabeled element is among the $k$ oldest elements in $lQueue$ and therefore the argument of one of the $k$ oldest $Enqueue$ instances in $\rho$, and $\rho \cdot op$ is legal.

    If $ret = \bot$, then there were no unlabeled elements in $lQueue$ with timestamp smaller than $ts(op)$, which means there are fewer than $k$ unmatched $Enqueue$ instances in $\rho$, and $\rho \cdot op$ is legal.

    By Corollary~\ref{ooolem:localExecOrder}, every process locally executes all $Dequeue$ instances in the same order, applying the same deterministic logic.  The only possible differences are that any $p_j$ may have fewer elements labeled for itself, as it has returned them to fast $Dequeue$ instances which it has not yet locally executed.  These instances do not affect the choice of return value for a slow $Dequeue$, though, so each $p_j$ will delete the same $ret$ when locally executing $op$.  

    Finally, we need to prove that Claim~\ref{ooolem:safeLabels} still holds after each process executes line~\ref{oooline:label}, and that all processes label the same elements for $p_i$.  But this follows by the same logic that tells us that $ret$ was the argument of one of the first $k$ unmatched $Enqueue$ instances in $\rho$, as each process chooses elements to label that are the oldest unlabeled elements in $lQueue$.  Thus, each element labeled in the local execution of $op$ is the argument of one of the first $k$ unmatched $Enqueue$ instances in $\rho \cdot op$.  Since all processes locally executed all $Dequeue$ instances in the same order, each will label the same elements while locally executing $op$.
  \item $op$ is a fast $Dequeue$ instance, $op = Dequeue(,-,ret)$.  Since we only label elements for process $p_i$ during local execution of a slow $Dequeue$ instance $op_\ell$ invoked at $p_i$ and no process can have two pending operation instances at the same time, we know that $op$ is invoked after $op_\ell$ returned, so $op$ linearizes after $op_\ell$.  Thus, by Claim~\ref{ooolem:safeLabels}, $ret$ was the argument of one of the first $k$ unmatched $Enqueue$ instances in $\pi$ before $op_\ell$, so it must still be before $op$.  By Lemma~\ref{ooolem:labelOwner}, no other process could have returned $ret$, and $p_i$ removes $Dequeue$ return values from its $lQueue$ before returning them, so could not have returned it as the return value of another $Dequeue$ instance, and thus $ret$ is still the argument of one of the first $k$ unmatched $Enqueue$ instances in $\rho$.
  \end{itemize}
\end{proof}

\begin{theorem}
  Algorithm~\ref{alg:relaxed} implements a queue with $k$-Out-of-Order relaxed $Dequeue$.
\end{theorem}

\subsection{Complexity}

Since each operation instance returns to the user in or before local execution, which occurs no later than when the invoking process receives acknowledgments from every process, sent when they first receive a message about the invocation, the worst-case response time of each operation in Algorithm~\ref{alg:relaxed} is $2d$.  

However, in many typical scenarios, most $Dequeue$ instances will return with no delay at all.  Since we use the same labeling structure as the previous algorithm \cite{TalmageWelch14}, the runtime analysis will be the same, but with slow $Dequeue$ instances taking $2d$ time instead of $d+\epsilon$ and fast $Dequeue$ instances taking $0$ time instead of $\epsilon$.  Specifically, in a \emph{heavily-loaded} run, which starts with at least $k$ $Enqueue$ instances, and maintains at least $k$ more $Enqueue$ than $Dequeue$ instances in any prefix, only one in every $\lfloor k/n\rfloor$ $Dequeue$ instances will take $2d$ time to return, giving us the following theorem:

\begin{theorem}
  In any complete, admissible, heavily-loaded run, the total cost of $m$ $Dequeue$ instances, with $m_i$ at each process $p_i$ is at most $\sum_{i=0}^{n-1}\left\lceil\frac{m_i}{\lfloor k/n\rfloor}\right\rceil \leq \left(\frac{m-n}{\lfloor k/n\rfloor} + n\right)(2d)$.
\end{theorem}

Each consecutive sequence of $\lfloor k/n\rfloor$ $Dequeue$ instances at each process will take $2d$ total time, which gives the first term.  The worst case is if there is one final, slow $Dequeue$ instance at each process, giving the extra $n(2d)$.  As $m$ increases, the effect of such a final slow $Dequeue$ instance vanishes, giving a practical runtime bound of $\left(\frac{m}{\lfloor k/n \rfloor}\right)(2d)$.

This is inversely proportional to the degree of relaxation $k$, so we can see that increasing relaxation allows us to reduce the total cost at will.  At $k=2n$, we halve the total runtime.  At $k=3n$, it is a third of that of an unrelaxed queue, and so on.  Thus, relaxation allows us to practically circumvent lower bounds on the performance of shared data structures, via a tunable tradeoff between ordering guarantees and performance, in an asynchronous model just as in the previous partially synchronous model.

If the run is not heavily loaded, the average cost of each $Dequeue$ instance at a process in a group from a slow instance until immediately before the next will depend on the number of elements that initial slow $Dequeue$ finds to label for its invoking process.  This is an ``effective $l$'', where $l = \lfloor k/n\rfloor$ \cite{TalmageWelch14}, and we observe that that group of $Dequeue$ instances will have an average cost of $\frac{2d}{l'}$, where $l'$ is the number of elements the leading slow $Dequeue$ instance is able to label for its invoking process. 

Note that we did not include labeling elements on insertion before the first $Dequeue$ instance, as Talmage and Welch did.  Whether that optimization works in the asynchronous model is an open question.  It is certainly less straightforward, since having $Enqueue$ instance label elements is more dangerous with our more complicated linearization order.  In any case, our algorithm still labels as many elements as are available, up to $k$, at the first (and every) $Dequeue$ instance, so after at most a single expensive operation per process, the average runtime, with most $Dequeue$ instances instantaneous, holds.

It is also worth calling attention to the counterintuitive fact that fast $Dequeue$ instances are \emph{faster} in our asynchronous algorithm than in the previous partially-synchronous algorithm.  We would expect worse performance in the harder model, as we see for slow $Dequeue$ instances.   The previous algorithm needed non-zero time for fast $Dequeue$ instances so that it could use timestamp order for its linearization, so a more complex linearization, such as we use, could likely match this performance in the partially synchronous model, but this complete removal of communication-related cost from most $Dequeue$ instances is another contribution of our new algorithm.

\section{Conclusion}

We have presented two algorithms, one implementing a traditional FIFO queue and one implementing a queue with $k$-Out-of-Order relaxed $Dequeue$, in asynchronous failure-free systems.  Both of our algorithms have a worst-case complexity of one message round trip ($2d$) per operation instance.  The relaxed queue algorithm has the same worst-case complexity, but much lower average complexity.  Thus, we have shown that relaxation can not only improve a queue's performance in idealized partially-synchronous models as prior work showed, but also in asynchronous models, by making many $Dequeue$ instances able to return after purely local computation, leaving coordination to happen in the background and allowing the user program to make progress.

Both of our algorithms offer full replication.  There is an existing centralized algorithm for the unrelaxed queue with equal performance to ours \cite{Lynch96}, but that puts a heavy load and reliance on one process.  While we do not here consider failures, having a replicated algorithm puts us in a good place to extend our algorithms to tolerate failures.  We are continuing in that direction, exploring how much extra cost failure-tolerance imposes, and how much of that cost relaxation may enable us to avoid.

\bibliography{refs.bib}

\end{document}